\documentclass[11pt]{article}

\usepackage{sn-preamble} 
\usepackage{tikz}
\usepackage{verbatim}
\usepackage{cancel}
\usetikzlibrary{shapes.geometric}
\usetikzlibrary{patterns}

\newcommand{\ignore}[1]{}

\newcommand{\Shift}{\textsc{Shift-Tester}}
\newcommand{\z}[1]{z^{(#1)}}
\newcommand{\dtv}{\mathrm{dist}_{\mathrm{TV}}}
\newcommand{\Dyes}{\calD_{\mathrm{yes}}}
\newcommand{\Dno}{\calD_{\mathrm{no}}}
\newcommand{\yes}{\mathrm{yes}}
\newcommand{\no}{\mathrm{no}}

\title{Testing Sumsets is Hard \vspace{0.5em}}
\author{
Xi Chen\\ 
{Columbia University}\\
\href{mailto:xc2198@columbia.edu}{\texttt{ xc2198@columbia.edu}}
\and 
Shivam Nadimpalli\\
{Columbia University}\\
\href{mailto:sn2855@columbia.edu}{\texttt{sn2855@columbia.edu}}
\and 
Tim Randolph\\
{Columbia University}\\
\href{mailto:t.randolph@columbia.edu}{\texttt{t.randolph@columbia.edu}}
\and 
Rocco A. Servedio\\
{Columbia University}\\
\href{mailto:ras2105@columbia.edu}{\texttt{ras2105@columbia.edu}}
\and 
Or Zamir\\
{Tel Aviv University}\\
\href{mailto:orzamir@tauex.tau.ac.il}{\texttt{orzamir@tauex.tau.ac.il}}
\vspace{0.5em}
}
\date{}

\begin{document}

\pagenumbering{gobble}
\maketitle  

\begin{abstract}
    A subset $S$ of the Boolean hypercube $\F_2^n$ is a \emph{sumset} if $S = \{a + b : a, b\in A\}$ for some $A \subseteq \F_2^n$. Sumsets are central objects of study in additive combinatorics, featuring in several influential results.
    We prove a lower bound of $\Omega(2^{n/2})$ for the number of queries needed to test whether a Boolean function $f:\F_2^n \to \{0,1\}$ is the indicator function of a {sumset}. Our lower bound for testing sumsets follows from sharp bounds on the related problem of \emph{shift testing}, which may be of independent interest. We also give a near-optimal {$2^{n/2} \cdot \poly(n)$}-query algorithm for a smoothed analysis formulation of the sumset \emph{refutation} problem.
\end{abstract}

\newpage
\pagenumbering{arabic}

%!TEX root = ../testing-sumsets.tex

\section{Introduction}
\label{sec:intro}

In recent years, theoretical computer science has increasingly been influenced by ideas and techniques from \emph{additive combinatorics}, a field sitting at the intersection of combinatorics, number theory, Fourier analysis, and ergodic theory. Notable examples of this connection include communication complexity~\cite{cfl,rank-comm,ghz}, constructions of randomness extractors~\cite{Bourgain05,biw,rao2007exposition,dvir}, and property testing~\cite{Green05,sam-gowers}; we also refer the reader to various surveys on additive combinatorics from the vantage point of theoretical computer science~\cite{additive-comb-minicourse,Trevisan09,Viola11ac,Bibak13,Lovett17actcs}.

Among the simplest objects of study in additive combinatorics are \emph{sumsets}: A subset $S$ of an abelian group $G$ (with group operation ``$+$'') is said to be a \emph{sumset} if $S=A+A$ for some $A \subseteq G$, where for sets $A,B \subseteq G$ we write $A+B$ to denote the set $\{a+b: a \in A, b \in B\}$.
Sumsets play a major role in additive combinatorics, where their study has led to many questions and insights about the additive structure of subsets of abelian groups. They are the subject of touchstone results in the field such as Freiman's theorem \cite{freiman1973foundations}, which (roughly speaking) says that if  $|A + A|$ is ``not too much larger'' than $|A|$ then $A$ must be contained in a generalized arithmetic progression which is ``not too large.'' 

Our interest in this paper is in \emph{algorithmic} questions related to sumsets.  There has been some prior research on sumset-related problems in the context of exact and parameterized algorithms. In \cite{fagnot2009finding} Fagnot, Fertin, and Vialette  considered the \textsc{2-Sumset Cover} problem: given a set $S$ of integers, does there exist a set $A$ of cardinality at most $k$ such that $S \subseteq A+A$? They proved $APX$-hardness for this problem and presented a $\poly(k) \cdot 5^{k^2(k+3)/2}$-time algorithm. The latter~was  improved to $\poly(k) \cdot 2^{(3 \log k - 1.4)k}$ by Bulteau, Fertin, Rizzi, and Vialette \cite{bulteau2015some}. \textsc{2-Sumset Cover} itself specializes \textsc{Generating Set}, in which the goal is to find a minimal set $A$ such that $S \subseteq \{\sum_{i \in I} i\; ; \; I \subseteq A\}$ and which was studied in \cite{collins2007nonnegative}. Given $S$ and $k$, finding a set $A$ of size $|A| \geq k$ with $A + A \subseteq S$ is equivalent to finding a $k$-Clique on the Cayley sum graph of $S$; this problem remains NP-hard, but can be solved with existing algorithms for $k$-clique \cite{godsil2017hardness}.

Apart from the works mentioned in the previous paragraph, computational questions pertaining to sumsets---{and more broadly, to objects and properties arising naturally in additive combinatorics}---do not appear to have been well studied to date.  In the list \cite{CrootLev07} of open problems in additive combinatorics compiled by Croot and Lev fifteen years ago, only one of a dozen questions concerning sumsets deals with algorithmic aspects (Problem~4.11, ``Recognizing sumsets algorithmically''), and the only subsequent reference to this problem that we have been able to identify is a 2020 MathOverflow post \cite{2020mathoverflowsumsets}. In \cite{AGU10}, Alon, Granville and Ubis ask whether it is possible to ``quickly identify'' subsets of $\F_p$ which are of the form $A+B$ with $|A|,|B| \geq k$, or  which have ``many representations'' as $A+B$, but we are not aware of subsequent research addressing these questions. We find this lack of study to be surprising given the increasingly prominent role of additive combinatorics in theoretical computer science (as before, see the expositions and surveys by \cite{Trevisan09,Viola11ac,Bibak13,Lovett17actcs}) and the fact that many such questions are both simple and natural.

\subsection{This Work}

In this paper we restrict our attention to the case in which the ambient abelian group $G$ is $\F_2^n$.  We do this for several reasons: For one thing, our focus is on algorithmic problems, and $\F_2^n$ is a very natural domain to consider from a theoretical computer science perspective.  Another motivation is that $\F_2^n$ is in some sense the simplest setting for many problems involving sumsets;  as Green stated in \cite{Green05}, ``the reason that finite field models are nice to work with is that one has the tools of linear algebra, including such notions as subspace and linear independence, which are unavailable in general abelian groups.'' (Indeed, several of our arguments will make use of these linear-algebraic tools.)

Since $\F_2^n$ is an exponentially large domain, it is natural to approach the study of sumsets over $\F_2^n$ from the vantage point of \emph{sublinear} algorithms. Thus, we will be interested in algorithms for which either the running time or the number calls to an oracle for the input set $S$ (i.e.~queries of the form ``does element $x$ belong to the set $S$?'' is less than $2^n$.  The recent work \cite{DNS22soda} took such a sublinear-algorithms perspective; it studied a problem which was closely related to the problem of approximating the size of the sumset $A+A$, given access to an oracle for the unknown set $A \subseteq \F_2^n$.  The main result of \cite{DNS22soda} was that in fact $O_\eps(1)$ queries---in particular, with no dependence on the dimension parameter $n$---are sufficient for a $\pm \eps \cdot 2^n$-accurate approximation of the quantity that they consider.  This naturally motivates the following broad question: What other algorithmic problems involving sumsets may be solvable with ``constant'' (depending only on $\eps$) or very low query complexity?

Motivated by this general question, in the current work we study a number of algorithmic questions related to sumsets.  The main problems we consider are described below:

\begin{enumerate}

\item We study (approximate) sumset recognition from a property testing perspective.    In more detail, given access to a membership oracle for an unknown set $S \subseteq \F_2^n$, in the \emph{sumset testing problem} the goal is to output ``yes'' with high probability (say, at least 9/10) if $S$ is a sumset and ``no'' with high probability if $S$ is $\eps$-far from every sumset (i.e.~$|S \, \triangle \,(A+A)| \geq \eps 2^n$ for every set $A \subseteq \F_2^n$), while making as few queries to the oracle as possible.

\item The above-described sumset testing problem turns out to be closely related to the problem of \emph{shift testing}, which is defined as follows: A shift testing algorithm is given black-box access to two oracles ${\cal O}_A,{\cal O}_B: \F_2^n \to \{0,1\}$, which should be viewed as membership oracles for two subsets $A,B \subseteq \F_2^n$.  The algorithm must output ``yes'' with probability at least $9/10$ if $B = A + \{z\}$ for some string $z \in \F_2^n$ and must output ``no'' with probability at least $9/10$ if the symmetric difference $B \, \triangle \, (A+\{z\})$ has size at least $\eps 2^n$ for every $z \in \F_2^n$.
% \xnote{I feel that the problem is already clear in the description above so there is no need to define it again here. If it helps we can add ''$|B \, \triangle \, (A+\{z\})|\ge \eps 2^n$ for every $z\in \mathbb{F}_2^n$'' to the description above.}

\item For $S \subseteq \F_2^n$, let $\bN_\eps(S)$ denote a random set which is an ``$\eps$-noisy'' version of $S$, obtained by flipping the membership\hspace{0.06cm}/\hspace{0.06cm}non-membership of each $x \in \F_2^n$ in $S$ with probability $\eps$.  It can be shown that for every $S \subseteq \F_2^n$ and every constant $0<\eps<1$, the noisy set $\bN_\eps(S)$ is with high probability not a sumset.  We study the problem of algorithmically \emph{certifying} that $\bN_\eps(S)$ is not a sumset; i.e.~we are given access to a membership oracle for $\bN_\eps(S)$, where $S$ is an arbitrary and unknown subset of $\F_2^n$, and the goal is to output a set $C \subseteq \F_2^n$ of points such that there is no sumset $A+A$ for which $(A+A) \cap C = \bN_\eps(S) \cap C.$  We refer to this problem as the \emph{smoothed sumset refutation problem,} since it aligns with the well-studied framework of smoothed analysis \cite{ST01} in which an arbitrary worst-case instance is subjected to a mild perturbation.

\end{enumerate}

Our main results for these problems are as follows.

\medskip

\noindent {\bf Sumset Testing Lower Bound.}
We give an $\Omega(2^{n/2})$ lower bound on the query complexity of sumset testing:

\begin{theorem} \label{thm:sumset-testing-lower-bound}
There is a constant $\eps > 0$ (independent of $n$) such that any algorithm ${\cal A}$ for the $\eps$-sumset testing problem must make $\Omega(2^{n/2})$ oracle calls.
%
%Let ${\cal A}$ be any algorithm with the following performance guarantee:  ${\cal A}$ is given access to a membership oracle for an unknown and arbitrary subset $S \subseteq \F_2^n$, and (i) if $S$ is a sumset then with probability at least $9/10$ algorithm ${\cal A}$ outputs ``sumset,'' while (ii) if $S$ is $\eps$-far from every sumset then with probability at least $9/10$ ${\cal A}$ outputs ``far from sumset'' (here $\eps>0$ is some fixed constant\rnote{Do we want to say something eventually about what constant we can achieve?}).  Then ${\cal A}$ must make at least $\Omega(2^{n/2})$ oracle calls.
\end{theorem}

We note that \Cref{thm:sumset-testing-lower-bound} holds even for adaptive testers which may make two-sided error.

\medskip

\noindent {\bf Tight Bounds for Shift Testing.} We show that the query complexity of shift testing is ${\Theta}^*(2^{n/2})$.

\begin{theorem} \label{thm:shift-testing}
(1)
There is an algorithm for the shift testing problem which makes $O(n2^{n/2}/\eps)$ oracle calls and runs in time $\poly(n) \cdot 2^n/\eps$.  Moreover,
(2) For any constant $0<c<1/2$, any algorithm for the shift testing problem must make $\Omega(2^{n/2})$ oracle calls, even for $\eps = (1/2 - 1/2^{cn})$.
This lower bound holds even for distinguishing the following two cases:  (i) $A$ is a uniform random subset of $\F_2^n$ and $B=A+\{z\}$ for a uniform random $z \in \F_2^n$; versus (ii) $A$ and $B$ are independent uniform random subsets of $\F_2^n$.
\end{theorem}

Like \Cref{thm:sumset-testing-lower-bound}, the lower bound (Part (2)) of \Cref{thm:shift-testing} holds even for adaptive testers which may make two-sided error.

\medskip

\noindent {\bf A Near-Optimal Algorithm for ``Smoothed'' Sumset Refutation.}  Our final result is a near-optimal algorithm which certifies that any noisy set $\bN_\eps(S)$ is not a sumset:

\begin{theorem} \label{thm:smoothed-sumset-refutation}
(1) There is an algorithm ${\cal A}$ for the $\eps$-smoothed sumset refutation problem that queries a set $\bB$ of $2^{n/2} \cdot O(n^{1.5}/\epsilon^{1.5})$ points with the property that $\bB \cap \bN_\epsilon(S)$ is inconsistent with every sumset with probability $1 - o_n(1)$. Moreover,  (2) Any algorithm for the $\eps$-smoothed sumset refutation problem, for any $\eps > 0$ must make $\Omega(2^{n/2}/\sqrt{n})$ many oracle calls.

%(1) There is an algorithm ${\cal A}$ with the following performance guarantee:  
% ${\cal A}$ is given access to a membership oracle for a set $S' \sim N_\eps(S)$, where $S \subseteq \F_2^n$ is unknown and arbitrary.  ${\cal A}$ makes $BLAH$ oracle calls and runs in time BLAHBLAH, and with probability at least $9/10$ (over the random draw of $S' \sim N_\eps(S)$ and the internal randomness of ${\cal A}$), ${\cal A}$ outputs a set $T \subseteq \F_2^n$ such that $T \cap N_\eps(S) \neq T \cap (A+A)$ for every $A \subseteq \F_2^n$.
%Moreover, (2) Any algorithm with the above performance guarantee must make $\Omega(BLAH)$ many oracle calls.
\end{theorem}

\subsection{Technical Overview} \label{subsec:techniques}

The main idea of our algorithm for shift testing (Part (1) of \Cref{thm:shift-testing}) is to query one oracle with all shifts of a random point $r$ by a subspace $V$, and query the other oracle by all shifts of the same point $r$ by the orthogonal complement $V^\perp$. This requires only $O(2^{n/2})$ queries, while providing information about the relationship between the two oracles vis-a-vis any possible shift $z \in \F_2^n$, since every possible shift has a decomposition into $z=z_1 + z_2$ for some $z_1 \in V, z_2 \in V^\perp.$

The optimality of this general approach is witnessed by the lower bound in Part (2) of \Cref{thm:shift-testing}.  The proof is by a ``deferred decisions'' argument which analyzes the knowledge transcript of a query algorithm which may be interacting either with the ``yes''-pair of oracles or the ``no''-pair of oracles.  We describe a coupling of the knowledge transcripts between these two cases, and use it to argue that if fewer than $0.1 \cdot 2^{n/2}$ queries have been made, then with high probability the transcripts are identically distributed across these two cases.  (See \cite{GoldreichRon02} for a similar high-level argument, though in an entirely different technical setting.)

The $\Omega(2^{n/2})$ lower bound of \Cref{thm:sumset-testing-lower-bound} for sumset testing is by a reduction to the lower bound for shift testing. We give a straightforward embedding of the ``$A$ is random, $B=A+\{z\}$''-versus-``$A,B$ are independent random'' shift testing problem over $\F_2^n$ into the problem of sumset testing over $\F_2^{n+2}$. The most challenging part of the argument is to prove that in fact the ``no'' instances of shift testing (when $A,B$ are independent random sets) give rise to instances which are far from sumsets over $\F_2^{n+2}$. This requires us to argue that a subset of $\F_2^{n+2}$ which is constant on two $n$-dimensional cosets and is uniform random on the other two $n$-dimensional cosets, is likely to be far from every sumset, which we prove using a linear algebraic argument.

For \Cref{thm:smoothed-sumset-refutation}, a result due to Alon establishes that every subset of the Boolean cube of size $2^n - c 2^{n/2} / \sqrt{n}$ is a sumset for a small constant $c$, which implies that any sumset ``$0$-certificate'' has size $\Omega^*(2^{n/2})$ \cite{Alon07}. To find such a certificate, we show that few noisy sumsets are likely to be consistent with an arbitrarily chosen subspace of dimension $n/2$, and then use a small random sample to rule out these sumsets with high probability.

\subsection{Discussion}
\label{subsec:related-work}

Our results suggest many questions and goals for future work; we record two such directions here.

The first direction is to obtain stronger results on sumset refutation.  
%A first goal in this regard might be to close (or narrow) the gap between our $\Omega(2^{n/2}/\sqrt{n})$ lower bound and our $2^{\tilde{O}(n/2)}$ upper bound on the size of ``certificates'' that noisy sets $\bN_\eps(S)$ are not sumsets.  Perhaps a more interesting goal is to
Is it possible to strengthen our sumset refutation result  by eliminating the ``smoothed analysis'' aspect, i.e.~is it the case that any $S\sse \F_2^n$ that is $\eps$-far from every sumset has a ``$0$-certificate'' of size~$2^{n/2} \cdot \poly\left(n, \frac{1}{\varepsilon}\right)$? If so, can such certificates be found efficiently given query access to $S$?

The second, and perhaps most compelling, direction is to either strengthen our $\Omega(2^{n/2})$-query lower bound, or prove an  upper bound, for the sumset testing problem.  We are cautiously optimistic that the true query complexity of sumset testing may be closer to $2^{n/2}$ queries than to $2^n$ queries, but any nontrivial ($o(2^n)$-query) algorithm would be an interesting result. One potentially relevant intermediate problem towards sumset testing is the problem of \emph{$k$-shift testing}, in which the goal is to determine whether oracles $\calO_A, \calO_B : \F_2^n\to\zo$ correspond to $B  = A + \{s_1, \ldots, s_k \}$ for some $k$ ``shift'' vectors $(s_{i})_{i\in[k]}$ versus $B$ being $\eps$-far from every union of $k$ shifts of $A$.

%!TEX root = ../testing-sumsets.tex

\section{Preliminaries}
\label{sec:prelims}

All probabilities and expectations will be with respect to the uniform distribution, unless otherwise indicated. We use boldfaced characters such as $\bx, \boldf$, and $\bA$ to denote random variables (which may be real-valued, vector-valued, function-valued, or set-valued; the intended type will be clear from the context).
We write $\bx \sim \calD$ to indicate that the random variable $\bx$ is distributed according to the probability distribution $\calD$. We write $\dtv(\calD_1,\calD_2)$ to denote the \emph{total variation distance} or \emph{statistical distance} between the distributions $\calD_1$ and $\calD_2$.

For $\epsilon \in [0, 1]$, we write $\bR_\varepsilon$ to denote a random subset of $\F_2^n$ obtained by selecting each element with probability $\epsilon$, so the ``$\eps$-noisy version'' of a set $S \subseteq \F_2^n$, denoted $\bN_\eps(S)$, is equivalent to $S \triangle \bR_{\eps}$,
where $A \triangle B := (A\setminus B) \cup (B\setminus A)$ denotes the symmetric difference of $A$ and $B$.

Given a set $A\sse\F_2^n$, we will write $\calO_A: \F_2^n\to\zo$ to denote the membership oracle for $A$, i.e. 
\[\calO_A(x) = \begin{cases}
 1 & x\in A\\ 0 	& x\notin A
 \end{cases}\] 
for $x\in\F_2^n$. Given $A, B\sse\F_2^n$, we write $\dist(A,B)$ for the normalized Hamming distance between the sets $A$ and $B$, i.e. 
\[\dist(A,B) := \frac{|A\triangle B|}{2^n} = \Prx_{\bx\sim\F_2^n}\sbra{\calO_A(\bx)\neq\calO_B(\bx)}.\]
We will also write 
$A + B := \cbra{ a+ b : a\in A, b\in B}.$
If one of the sets is a singleton, e.g. if $A = \{a\}$, we will sometimes write $a+B := \{a\} + B$ instead. 

We write $H(x)$ to denote the binary entropy function $-x \log_2 x - (1-x) \log_2(1-x)$. Stirling's approximation gives us the following helpful identity:
\begin{equation}\label{eq:stirling}
    \binom{n}{\alpha n} = \Theta^*(2^{H(\alpha)n}); \text{ or, equivalently,  }
    \binom{2^n}{\alpha 2^n} = 2^{H(\alpha)2^n} \cdot 2^{\Theta(n)}.
\end{equation}

Given a subset $D$ of an Abelian group $G$, we write $\Gamma_G(D)$ to denote the \emph{Cayley sum graph} of $G$ with respect to the generator set $D$; that is, the graph on the vertex set $G$ that contains the edge $(x, y)$ if and only if $x + y \in D$. (Since the group we consider is $\F_2^n$, for us this is the same as the regular Cayley graph of $G$ with respect to generator set $D$.) When $D=\{x\}$ is a singleton for some $x \in G$, we abuse notation slightly and write $\Gamma_G(x)$ for $\Gamma_G(\{x\})$.

% \begin{definition}
%     For a set~$D\subseteq \F_2^n$ denote by~$\Gamma\left(\F_2^n,D\right)$ the graph with vertex set~$\F_2^n$ that contains each edge~$(x,y)$ if and only if~$x+y\in D$.
%     This is usually denoted as the Cayley graph of the group~$\F_2^n$ with respect to the generator set~$D$.
% \end{definition}
%!TEX root = ../testing-sumsets.tex

\section{Optimally Testing Shifts}
\label{sec:shifts}

Given $A, B \subseteq \F_2^n$, we say that $B$ is a \emph{shift} of $A$ if there exists $z\in\F_2^n$ such that $A + z = B$. We obtain the following upper and lower bounds for the \emph{shift testing} problem:

\begin{theorem} \label{thm:shifts}
    Let $\calO_A, \calO_B : \F_2^n\to\zo$ be membership oracles for $A,B\sse\F_2^n$. Then:
    \begin{enumerate}
        \item The algorithm \Shift~(\Cref{alg:shift-testing}) makes $O(n2^{n/2}/\eps)$ oracle calls, runs in time $\poly(n) \cdot 2^n/\eps$, and guarantees that:
        \begin{enumerate}
    	\item If $B = A+z$ for some $z\in\F_2^n$, the algorithm outputs ``shift'' with probability $9/10$;
    	\item If for every $z\in\F_2^n$ we have $\dist(A + z, B) \geq \eps$, the algorithm outputs ``$\eps$-far from shift'' with probability $9/10$.
        \end{enumerate}
        \item Fix $c$ to be a constant that is less than $1/2$. Any (adaptive, randomized) algorithm with the performance guarantee in the previous item makes $\Omega(2^{n/2})$ oracle calls, even for $\eps = 1/2 - 1/2^{cn}.$ 
        % \tnote{ Silly note, but I'm confused by this wording: is this equivalent to saying ``for any constant $\epsilon < 1/2 - 1/2^{n/2}$''? {\bf Rocco:} To me saying ``for any constant $\epsilon < 1/2 - 1/2^{n/2}$'' means the same thing as ``for any constant $\eps < 1/2$'', i.e. 0.49 is okay but $1/2 - 1/n$ is not.  The current phrasing is meant to capture that $1/2 - 1/2^{0.49n}$ is okay but $1/2 - 1/2^{0.51n}$ is not. Does this make sense? }
    \end{enumerate}
\end{theorem}

In fact, the lower bound holds even for distinguishing the following two cases:  (i) $A$ is a uniform random subset of $\F_2^n$ and $B=A+s$ for a uniform random $s \in \F_2^n$; versus (ii) $A$ and $B$ are independent uniform random subsets of $\F_2^n$.

\subsection{Upper Bound}
\label{subsec:shift-ub}

In this section, we prove Item 1 of \Cref{thm:shifts}. Note that since
\[\dist(B,A+z) = \Prx_{\bx}\left[\calO_A(\bx) \neq \calO_B(\bx+z)\right],\]
if $B$ is a shift of $A$ (i.e. $B = A + z_\ast$ for some $z_\ast$), we then have for that $z_\ast$ that
\[\Prx_{\bx}\left[\calO_B(\bx) = \calO_A(\bx+z_\ast)\right]=1.\]
On the other hand, if $\dist(B,A+z) \ge \eps$ for every $z$, then for every $z$ we have  
\[\Prx_{\bx}\left[\calO_B(\bx) = \calO_A(\bx+z)\right]\le 1-\eps.\]
These simple observations suggest that in order to estimate $\Pr[\calO_B(\bx) = \calO_A(\bx + z)]$ for a particular $z$, we would like to make queries $\calO_B(\bx),\calO_A(\bx+z)$ for uniform random $\bx.$  The fact that we need to do this for all $z$ motivates the following approach; before proceeding, we introduce some notation.

\begin{notation}
	We define the subsets $D_1,D_2 \subset \F_2^n$,  where $D_1$ is the set of all $2^{\lfloor n/2 \rfloor}$ vectors whose last $\lfloor n/2 \rfloor$ coordinates are all-0 and $D_2 \subset \F_2^n$ is the set of $2^{\lceil n/2 \rceil }$ vectors whose first $
 \lceil n/2 \rceil$ coordinates are all-0.
Note that every $z \in \F_2^n$ has a unique expression as 
\[
z:=\z{1} + \z{2},\qquad\text{for}~ \z{1} \in D_1 ~\text{and}~ \z{2} \in D_2.
\]
%(This is the essential property we require in what follows; indeed, we could have chosen $D_1,D_2$ to be any two disjoint $2^{n/2}$-size subsets for which this holds.)
% TR: I think the above is implicit; fine to uncomment
\end{notation}

Fix a particular string $z=\z{1} + \z{2}$ as above.  
%If $\br$ is uniform over $\F_2^n$ then so is $\br + \z{1}$; 
We write $\bx=\br+\z{1}$, and we observe that if $\br$ is uniform random then so is $\bx$. 
As alluded to earlier we would like to query $B$ on $\bx$ and $A$ on $\bx + z = \br + \z{1} + z =\br + \z{2}$.  
The main observation is that if we query $B$ on all strings in $D_1 + \br$ and query $A$ on all strings in $D_2 + \br$, then no matter what $z$ is we will have made the queries $\calO_B(\bx)=\calO_B(\br + \z{1})$ and $\calO_A(\bx + z) = \calO_A(\br + \z{2})$, so we will have obtained a sample towards estimating $\Pr_{\bx}[\calO_B(\bx) = \calO_A(\bx + z)]$.
Since this is true for every $z$, we can reuse the above queries towards all possibilities for $z$.  
(Of course one sample is not enough to estimate a probability, so we will repeat the above with $n/\eps$ different choices of $\br$.)

\begin{algorithm}
    \vspace{1em}
    \textbf{Input:} Oracles $\calO_A, \calO_B :\F_2^n \to \zo$ and $\epsilon > 0$ \\[0.25em]
    \textbf{Output:} ``Shift'' or ``$\eps$-far from shift''

    \vspace{1em}

    \textsc{\Shift$(\calO_A, \calO_B, \epsilon)$:}
    \begin{enumerate}
        \item Repeat the following $n/\epsilon$ times:
        \begin{enumerate}
        		\item Draw a uniformly random $\br\in\F_2^n$.
        		\item Query $\calO_A$ on all $x \in D_1 + \br$, and query $\calO_B$ on all $y \in D_2 + \br$. 
        	\end{enumerate}
        	\item For each $z = z^{(1)} + z^{(2)}\in \F_2^n$, let $p_z$ be the fraction of the $n/\eps$ repetitions for which  
        	\[\calO_A(\br + z^{(1)}) = \calO_B(\br + z^{(2)}).\]
        	\item If $p_z = 1$ for some $z\in\F_2^n$, output ``shift''; otherwise output ``$\eps$-far from shift''.
    \end{enumerate}
    \caption{An algorithm for shift testing.}
    \label{alg:shift-testing}
\end{algorithm}

\begin{proofof}{Item 1 of \Cref{thm:shifts}}
Our algorithm, $\Shift$, is presented in \Cref{alg:shift-testing}.
	Note that if $B$ is a shift of $A$, i.e. if there exists a $z_{\ast}\in\F_2^n$ for which $B = A + z_{\ast}$ then 
	\[\br + \z{1}_{\ast} \in A \qquad\text{if and only if}\qquad \br + \z{1}_{\ast} + z_{\ast} = \br + \z{2}_\ast \in B,\]
	where we used the fact that $\z{1}_{\ast} + z_{\ast} = \z{1}_{\ast} + \z{1}_{\ast} + \z{2}_{\ast} = \z{2}_{\ast}$. In particular, we will have $p_{z_\ast} = 1$ and so the algorithm will return ``shift'' with probability $1$. On the other hand, suppose $B$ is $\eps$-far from $A+z$ for every $z\in \F_2^n$; fix any such $z$. Then the probability that all $n/\eps$ repetitions in \Cref{alg:shift-testing} will have $\calO_A(\br + \z{1}) = \calO_B(\br + \z{2})$ is at most 
	\[(1-\eps)^{n/\eps} \leq e^{-n}.\]
	Taking a union bound over all $z\in\F_2^n$ implies that the probability that \Cref{alg:shift-testing} will output ``$\eps$-far from shift'' is at least $1 - (2/e)^n$, completing the proof.
\end{proofof}

Note that \Cref{alg:shift-testing} in fact has a stronger guarantee than what is required by \Cref{thm:shifts}: The algorithm never outputs ``$\eps$-far from shift'' if $B$ is a shift of $A$, and if $B$ is $\eps$-far from every shift of $A$ then the algorithm outputs ``shift'' with probability at most $(2/e)^n$.

\subsection{Lower Bound}
\label{subsec:shift-lb}

To prove Item 2 of \Cref{thm:shifts} we define two probability distributions, $\Dyes$ and $\Dno$, over instances of the shift testing problem.

\begin{definition} \label{def:shift-dyes}
    A draw $(\bA,\bB)$ from $\Dyes$ is obtained 
    as follows:
    \begin{itemize}
        \item $\bA \subseteq \F_2^n$ includes each element of $\F_2^n$ independently with probability 1/2.
        \item $\bB \sse \F_2^n$ equals $\bA + \bs$ for $\bs$ sampled uniformly at random from $\F_2^n$.
    \end{itemize}
\end{definition}

Note that for $(\bA,\bB)\sim \Dyes$, $\bB$ is a shift of $\bA$.   

\begin{definition} \label{def:shift-dno}
A draw $(\bA,\bB)$ from $\Dno$ is obtained as follows: 
\begin{itemize}
\item $\bA \subseteq \F_2^n$ includes each element of $\F_2^n$ independently with probability 1/2. 

\item $\bB \subseteq \F_2^n$ also includes each element of $\F_2^n$ with probability 1/2 (independently of $\bA$).

\end{itemize}
\end{definition} 

A straightforward application of the Chernoff bound, combined with a union bound over the $2^n$ possible shifts, shows that with probability at least $19/20$ a draw of $(\bA,\bB) \sim \Dno$ is such that $\bB$ is $(1/2 - 1/2^{cn})$-far from every shift of $\bA$ (for any constant $c<1/2$). So to prove Item 2 of  \Cref{thm:shift-testing}, it is enough to establish the following claim for deterministic algorithms. (By Yao's minimax principle, this is sufficient to prove a lower bound for randomized algorithms as well.) 

\begin{claim} \label{claim:indistinguishable}
Let ${\tt Test}$ be any deterministic,
%\xnote{I think it suffices to prove it for deterministic algorithms.}, 
adaptive algorithm that makes $N:=0.1 \cdot 2^{n/2}$ oracle calls to ${\cal O}_A$ and ${\cal O}_B$. Let $T_{\tt test}(A,B)$ be the ``transcript'' of its queries to the oracles and received responses, i.e. $T_{\tt test}(A,B)$ consists of

\begin{center}
(first query to one of the oracles, response received)

$\vdots$

($N$-th query to one of the oracles, response received).
\end{center}
Then we have 
\[\dtv\left(T_{\tt test}(\bA_\yes,\bB_\yes),T_{\tt test}(\bA_\no,\bB_\no)\right) \leq 0.02,\]
where $(\bA_\yes,\bB_\yes) \sim \Dyes$ and $(\bA_\no,\bB_\no) \sim \Dno$.
\end{claim}

The claim follows by analyzing the behavior of the algorithm on an oracle constructed over the course of answering the queries posed by algorithm {\tt Test}, i.e., ``deferring'' the decision of whether the oracle $(\bA,\bB)$ is drawn from $\Dyes$ or $\Dno$. See for example Section~7.1 of \cite{GoldreichRon02}.  
% \rnote{cf.~p.~26-29 of \href{https://www.wisdom.weizmann.ac.il/~oded/R2/bdg.pdf}{https://www.wisdom.weizmann.ac.il/~oded/R2/bdg.pdf} for a (very) detailed example of this kind of argument. Maybe we want to be a little more formal/rigorous than what I wrote below, but I think we may not want to reach the level of detail that is that paper.}

\def\Test{{\tt Test}}

\begin{proof}
For simplicity, we assume that in each round, $\Test$ queries one point $q$ and receives \emph{both} $\calO_A(q)$ and $\calO_B(q)$; this can only make $\Test$ more powerful.
%\xnote{Test could potentially query the same point twice, once for $A$ and once for $B$. For convenience let's give it both $A(q)$ and $B(q)$ when $q$ is queried.}

Consider the following approach to answering queries posed by {\tt Test}:  before any queries are made, draw a uniform random $\bs \sim \F_2^n$.  Let $q_1,\dots,q_{t-1} \in \F_2^n$ be the first $t-1$ queries made by {\tt Test} (we may suppose without loss of generality that all these $t-1$ query strings are different from each other, since any algorithm that repeats a query string can easily be modified so as not to do so).  When the $t$-th query string $q_t$ is provided by {\tt Test}, the answer is generated as follows:
\begin{flushleft}
\begin{enumerate}
\item If $\bs \neq q_t + q_{t'}$ for all $t' \leq t$, then two independent uniform random bits $\bb_{\bA},\bb_{\bB} \in \{0,1\}$ are drawn and returned as $\calO_{\bA}(q_t)$  and $\calO_{\bB}(q_t)$. (It may be helpful to think of this outcome as being ``recorded'', i.e.~ when this happens the process ``decides'' that $\bb_{\bA},\bb_{\bB}$ are the values of $\bA$ and $\bB$ on the point $q_t$.)

\item If $\bs = q_t + q_{t'}$ for some $t' \leq t$, then the process halts and outputs ``failure.''
\end{enumerate}
\end{flushleft}

The key observation is that conditioned on the above process proceeding through $t$ queries without an output of ``failure'', the length-$t$ transcript is distributed exactly according to the pair of oracles $(\bA,\bB)$ being $(\bA_\yes,\bB_\yes) \sim \Dyes$, and also exactly according to the pair of oracles being $(\bA_\no,\bB_\no) \sim \Dno$.
This is because in either case, as long as no pair of queries $q_t,q_{t'}$ sum to the ``hidden'' random string $\bs \in \F_2^n$, every response to every oracle call is distributed as an independent uniform random bit.
%\xnote{I think the argument before feels that it is for algorithms that are nonadaptive so I give more details below.}   

We finish the proof by showing that
%, for any $N$-query deterministic, adaptive algorithm $\Test$, 
the probability that the process above outputs ``failure'' is at most $0.02$.
To this end, we note that conditioning on no ``failure'' during the first $t-1$ rounds $q_1,\ldots,q_{t-1}$, $\bs$ is distributed uniformly among all points in $\mathbb{F}_2^n$ that are not equal to $q_i+q_j$ for some $i,j\in [t-1]$.
%Since $N=0.1 \cdot  2^{n/2}$, the number of pairs $\{t,t'\}$ with $t,t' \leq N$ is less than $0.05 \cdot 2^n$, and hence the probability that $\bs = q_t + q_{t'}$ for any such pair is less than $1/20$.
%So the variation distance $\dtv(T_{\tt test}(\bA_\yes,\bB_\yes),T_{\tt test}(\bA_\no,\bB_\no))$ is less than $1/20$, and the claim is proved. 
The number of such points is at least $2^n-N^2/2>0.99\cdot 2^n$. On the other hand, the process outputs ``failure'' in round $t$ if one of 
$q_1+q_t,\ldots,q_{t-1}+q_t$ is $\bs$, which happens with probability at most $N/(0.99\cdot 2^n)<0.2\cdot 2^{-n/2}$. It follows from a union bound on the $N$ rounds that the process outputs ``failure'' with probability at most $0.02$. This finishes the proof of the claim.
\end{proof}

%!TEX root = ../testing-sumsets.tex

\section{Lower Bound for Testing Sumsets}
\label{sec:sumset-lb}

\newcommand{\sham}{\mathrm{sham}}

In this section, we show that the lower bound for shift testing established in \Cref{subsec:shift-lb} implies a lower bound for the problem of testing sumsets. More formally, we prove the following:

\begin{theorem} \label{thm:sumset-lb}
    Let $\calO_S: \F_2^n \to \zo$ be a membership oracle for %an unknown and arbitrary 
    $S\sse\F_2^n$. 
    There is an absolute constant $\eps > 0.0125$ such that the following holds: Let $\calA$ be any (adaptive, randomized) algorithm with the following performance guarantee:
    \begin{enumerate}
        \item If $S = A+A$ for some $A\subseteq \F_2^n$, $\calA$ outputs ``sumset'' with probability $9/10$; and
        \item If $\dist(S, A+A) \geq \epsilon$ for all $A\subseteq \F_2^n$, $\calA$ outputs ``$\eps$-far from sumset'' with probability $9/10$.
    \end{enumerate}
    Then $\calA$ must make $\Omega(2^{n/2})$ calls to $\calO_S$.
\end{theorem}

% \xnote{I think the proof now is not really a reduction so I revised the sentence.}
The distributions we use to prove  \Cref{thm:sumset-lb} are based on
  the distributions $\Dyes$ and $\Dno$
  defined in \Cref{def:shift-dyes,def:shift-dno} for shift testing.
%via a randomized reduction from shift testing. 
Given $A, B \sse \F_2^n$, we define  $\calS(A,B)\sse\F_2^{n+2}$
%\tnote{Since $\calS$ is just a random set, should we change it the notation to $\bS$? {\bf Rocco:} I think here we're defining ${\calS}(A,B)$ for any sets $A,B$ (which need not be random) so there is no randomness, right - the current notation seems okay to me } 
as 
\begin{equation} 
    \label{eq:reduction-definition}
    \calS(A,B) := \cbra{x : x_1 = x_2 = 0} \sqcup \cbra{ (1, 0, a) : a \in A} \sqcup \cbra{(0,1, b) : b \in B},
\end{equation}
where the notation $(b_1, b_2, v)$ indicates that the bits $b_1$ and $b_2$ are concatenated with $v \in \F_2^n$ to create an element in $\F_2^{n+2}$.
\Cref{fig:shift-to-sumsets} illustrates the set $\calS(A, B)$.
%Consider the distributions $\Dyes$ and $\Dno$ defined in \Cref{def:shift-dyes,def:shift-dno} and 

\def\bSyes{\bS_{\yes}}
\def\bSno{\bS_{\no}}

We use $\Dno$ to 
define $\calS_{\no}$, a distribution over subsets of $\mathbb{F}_2^{n+2}$ as follows:
To draw $\bSno\sim \calS_{\no}$, we draw  
  $(\bA_{\no}, \bB_{\no}) \sim \Dno$ and set $\bSno=\calS(\bA_{\no},\bB_{\no})$. 
%For convenience, we add one ``extra'' point to $\calS_{\yes}$, defining this set as 
On the other hand, we use $\Dyes$ to define $\calS_{\yes}$ as follows:
To draw $\bSyes\sim \calS_{\yes}$, we   draw $(\bA_{\yes}, \bB_{\yes} = \bA_{\yes} + \bs) \sim \Dyes$ but add one ``extra'' point to $\bSyes$, defining it as:
 $
    \bSyes = \calS(\bA_{\yes}, \bB_{\yes}) \sqcup \{(1, 1, \bs)\}.
$
This will ensure that $\bSyes\sim \calS_{\yes}$ is likely to be a sumset (see \Cref{prop:sumset-yes} below).

\begin{figure}[t]
    \centering

    \begin{tikzpicture}
        % \node (blah) at (-4, 0) {$A, B \sse \F_2^n$};
        % \draw[|->] (blah) -- (-1, 0);
        \draw[-] (0, 1.5) -- (0, -1.5);
        \draw[-] (1.5, 0) -- (-1.5, 0);
        \node (B) at (-0.75, 0.75) {$B$};
        \node (A) at (0.75, -0.75) {$A$};
        \node (0) at (0.75, 0.75) {$0$};
        \node (1) at (-0.75, -0.75) {$1$};
        \node () at (-1.75, -1.75) {$(0,0)$};
        \node () at (1.75, -1.75) {$(1,0)$};
        \node () at (-1.75, 1.75) {$(0,1)$};
        \node () at (1.75, 1.75) {$(1,1)$};
        \fill[pattern=north west lines,pattern color=black!50] (-1.5, -1.5) -- (-1.5, 0) -- (0,0) -- (0, -1.5);
        \draw[-] (-1.5, -1.5) -- (-1.5, 1.5) -- (1.5, 1.5) -- (1.5, -1.5) -- (-1.5, -1.5);
    \end{tikzpicture}
    
    \caption{The set $\calS(A, B) \sse \F_2^{n+2}$. By \Cref{prop:sumset-yes}, for a typical $(\bA_{\yes},\bB_{\yes})$ drawn from $\Dyes$, adding a single point $(1,1,\bs)$ in the top right cell makes $S(\bA_{\yes},\bB_{\yes})$ into a sumset.}
    
    \label{fig:shift-to-sumsets}
\end{figure}

At a high level, the proof of \Cref{thm:sumset-lb} contains three steps: we show (1) that $\bSyes\sim \calS_{\yes}$ is a sumset with high probability (\Cref{prop:sumset-yes}), (2) that $\bSno\sim \calS_{\no}$ is $\epsilon$-far from being a sumset with high probability (\Cref{prop:sumset-no}), and (3) that oracles to $\bSyes\sim \calS_{\yes}$ and $\bSno\sim \calS_{\no}$ are too similar for an algorithm that makes few queries to tell the difference, where `similarity' is measured in terms of the total variation distance between distributions over transcripts (proof of \Cref{thm:sumset-lb}). The theorem then follows quickly from these three facts.

\begin{proposition} \label{prop:sumset-yes}
    With probability at least $1 - 2^{-\Omega(2^n)}$, $\bSyes\sim \calS_{\yes}$ is a sumset over $\F_2^{n+2}.$
\end{proposition}
\begin{proof}
Let $(A_{\yes},B_{\yes})$ be a pair of sets in the support of $\Dyes$ with $B_{\yes}=A_{\yes}+s$.
It is easy to verify that 
$\calS(A_{\yes}, B_{\yes}) \sqcup \{(1, 1, s)\}$ is equal to 
$C+C$ with
%    Consider the random set $\bC\sse\F_2^{n+2}$, defined as
    \[C:= \left\{0^n \right\} \sqcup \left\{(1,0,a) : a\in A_{\yes}\right\} \sqcup \left\{(1,1,s)\right\}.\]
as long as $A_{\yes}+A_{\yes}$ covers all of $\F_2^n$.
So it suffices to show that this holds with extremely high probability with a uniformly random set $\bA_{\yes}$. 

To see this, consider any fixed, nonzero element $z \in \F_2^n$. Without loss of generality, suppose that the first coordinate of $z$ is $1$.  We have
    \begin{align*}
        \Pr\left[z \notin \bA_{\yes}+\bA_{\yes}\right] &=
        \Pr\big[\text{for all~}y \in \F_2^n, \text{~either~}y \notin\bA_{\yes}\text{~or~}z+y \notin \bA_{\yes}\big]
    = (3/4)^{2^{n-1}},
    \end{align*}
    where the second equality holds because $\bA_{\yes}$ is a uniform random subset of $\F_2^n$ and $y,z+y$ are distinct elements (observe that the first coordinate of $y$ is 0 while the first coordinate of $z+y$ is 1).
    Since $\Pr[0^n\notin\bA_{\yes}+\bA_{\yes}]=\Pr[\bA_{\yes}$ is empty$]=(1/2)^{2^n}<(3/4)^{2^{n-1}}$, we get that each fixed element $z \in \F_2^n$ is missing from $\bA_{\yes}+\bA_{\yes}$ with probability at most $(3/4)^{2^{n-1}}$.  
    The claim follows from a union bound over the $2^n$ elements of $\mathbb{F}_2^n$.
    %Hence the expected size of $\F_2^n \setminus (\bA_{\yes}+\bA_{\yes})$ is at most $2^n \cdot (3/4)^{2^{n-1}}$, and the probability that $\F_2^n \setminus (\bA_{\yes}+\bA_{\yes})$ is non-empty is at most $2^n \cdot (3/4)^{2^{n-1}}$ by Markov's inequality. 
    %Thus $\calS(\bA_{\yes},\bB_{\yes})$ equals the sumset $\bC + \bC$ with probability $1 - 2^{-\Omega(2^n)}$.
\end{proof}

\begin{proposition} \label{prop:sumset-no}
    With probability at least $1 - o_n(1)$, $\bSno\sim \calS_{\no}$ is $0.0125$-far from every sumset.
\end{proposition}

We now complete the proof of \Cref{thm:sumset-lb} using \Cref{prop:sumset-yes,prop:sumset-no}. The proof of \Cref{prop:sumset-no} is deferred to \Cref{subsec:sumset-lb-no}.

\begin{proofof}{\Cref{thm:sumset-lb}}
Let $\calA$ be an algorithm for sumset testing on $\mathbb{F}_2^{n+2}$ that makes at most $N = 0.1 \cdot 2^{n/2}$ queries. As in the proof of \Cref{thm:shifts}, we let $T_{\calA}(S)$ denote the $N$-element transcript of $\calA$ given the oracle $\calO_S$ and take a ``deferred decision'' approach to prove that $\calA$ cannot distinguish between $\calS_{\yes}$ and $\calS_{\no}$ with high probability.

By \Cref{prop:sumset-yes,prop:sumset-no}, the probability that $\bSyes\sim \calS_{\yes}$ is 
%1/2^{n+2}$-far from 
not a sumset is $o_n(1)$, and the probability that $\bS\sim \calS_{\no}$ is $0.0125$-close to any sumset is $o_n(1)$. As a result, to prove \Cref{thm:sumset-lb} it suffices to 
show that 
\[
    \dtv\left(T_{\calA}\big(\bSyes), T_{\calA}(\bSno)\right) < 0.1 - o_n(1),
\]
where $\bSyes\sim \calS_{\yes}$ and $\bSno\sim \calS_{\no}$.

Consider the sham oracle $\calO_{\sham}$ that samples a point $\bs \in \mathbb{F}_2^n$ uniformly at random, then responds to queries as follows:
\begin{enumerate}
    \item If the query is a point $q_t$ for which $q_{t,1} = q_{t,2} = 0$, the oracle returns 1.
    \item If the query is $(1, 1, \bs)$, the oracle outputs ``failure''. Otherwise, if $q_{t,1} = q_{t,2} = 1$, it returns 0.
    \item If the query is a point $q_t$ such that $q_t + q_{t'} = (1,1,\bs)$ for some previously queried point $q_{t'}$, the oracle outputs ``failure''. Otherwise, it returns a random bit.
\end{enumerate}

We proceed to consider the behavior of $\calA$ given $\calO_{\sham}$, $\calO_{\bS_{\yes}}$, and $\calO_{\bS_{\no}}$. Conditioned on the event that $\calO_{sham}$ does not output ``failure'', $\calA$ always receives the answer `1' when querying a point with initial coordinates $(0, 0)$, always receives the answer `0' when querying a point with initial coordinates $(1, 1)$, and receives a random bit when querying a point with the initial coordinates $(0, 1)$ or $(1, 0)$. 
If, after the point of ``failure'', our oracle subsequently responds to queries consistently with the distribution $\calS(\bA_{\yes}, \bA_{\yes} + \bs)$, randomly determining membership in $\bA_{\yes}$ via deferred decision as necessary, the resulting distribution over transcripts is identical to that given oracle access to $\bS_{\yes}$. Likewise, if the oracle responds `0' on $(1, 1, s)$ and continues to return random bits on queries whose initial coordinates begin with $(0,1)$ or $(1, 0)$, the resulting distribution over transcripts is identical to that given oracle access to $\bS_{\no}$. We conclude that the distribution of $T_{\calA}(\sham)$, the transcript of $\calA$ given $\calO_{\sham}$, is identical to the distribution of transcripts given $\calO_{\bS_{\yes}}$ and $\calO_{\bS_{\no}}$ unless failure occurs.

Failure is unlikely for any algorithm $\calA$ that makes at most $N$ queries: With $N$ queries, the algorithm can rule out at most $N = O(2^{n/2})$ candidates for $\bs$ by querying points with the initial coordinates $(1, 1)$, and at most $N^2 = 0.01 \cdot 2^{n}$ candidates for $\bs$ by querying points with the initial coordinates $(0,1)$ and $(1, 0)$. Conditioned on no failure, the posterior distribution of $\bs$ is thus uniform over at least $(0.99 - o_n(1)) 2^{n}$ points. Thus subsequently querying a point discovers $s$ with probability at most 
\[\frac{N}{(0.99 - o_n(1)) 2^{n}} \leq \frac{0.2}{2^{n/2}}.\] 
Union-bounding over all $N$ rounds gives a failure probability of at most $0.02$.

We conclude that
\begin{align*}
    \dtv(T_{\calA}(\bS_{\yes}), T_{\calA}(\bS_{\no})) \leq 0.02 + o_n(1),
\end{align*}
and thus any algorithm that makes at most $N = 0.1 \cdot 2^{n/2}$ queries cannot answer correctly with probability 9/10.
\end{proofof}

\subsection{Proof of \Cref{prop:sumset-no}}
\label{subsec:sumset-lb-no}

We prove \Cref{prop:sumset-no} via a counting argument. The distribution $\calS_{\no}$ produces subsets of $\mathbb{F}_2^{n+2}$ of a specific form: these subsets contain every point in the subspace $\{x: x_1 = x_2 = 0\}$, no points in the coset $\{x: x_1 = x_2 = 1\}$, and have density roughly 0.5 on the cosets $\{x: x_1 = 0, x_2 = 1\}$~and $\{x: x_1 = 1, x_2 = 0\}$. We first bound the number of sumsets that are \emph{$\epsilon$-eligible} (roughly, ``close'') to any subset of this form (\Cref{prop:eligible-sumset-count}). Since there are relatively few subsets of $\F_2^{n+2}$ near any $\epsilon$-eligible sumset, we conclude that most subsets drawn from $\calS_{\no}$ are far from any sumset (\Cref{prop:sumset-no}).

\begin{remark}
It can be shown that the number of sumsets in $\F_2^{n+2}$ is at most $2^{2^{n+1} + O(n^2)}$ (this bound is implicit in the work \cite{Sargsyan15}, and for completeness we give a proof in \Cref{appendix:count}). However, this upper bound is not enough for us per se since the support of $\calS_{\no}$ is also of size $2^{2^{n+1}}$; hence we need to use the more refined notion of ``$\eps$-eligible'' sumsets mentioned above.
\end{remark}

In the remainder of this section, we make frequent reference to the volume of sets within the subspace $\{x: x_1 = x_2 = 0\}$ of $\mathbb{F}_2^{n+2}$ and its three cosets. Given a set $S\sse\F_2^{n+2}$ and a pair of bits $(b_1, b_2) \in \zo^2$, we define
\[
    \Vol_{b_1b_2}(S) := \frac{|S\cap \{x \in\F_2^{n+2}: x_1 = b_1, x_2 = b_2\}|}{2^n}.
\]
in order to simplify notation.

\begin{definition} \label{def:eligible-sumset}
Given $\eps > 0$, we say that a set $S \sse \F_2^{n+2}$ is an \emph{$\eps$-eligible sumset} if $S = A+A$ for some $A\sse\F_2^{n+2}$ and if the following holds:
\[\Vol_{00}(S) \geq 1-\eps \qquad\text{and}\qquad \Vol_{11}(S) \leq \eps.\]
\end{definition}

Roughly, the $\epsilon$-eligible sumsets are all those that might be close to $\bSno\sim \calS_{\no}$. 

\begin{proposition} \label{prop:eligible-sumset-count}
    For any $\eps$, the number of $\eps$-eligible sumsets in $\F_2^{n+2}$ is at most 
    \[
        \max\left\{ 2^{4 H(\epsilon) \cdot 2^n}, 2^{ (1 + 2H(\epsilon)) 2^n} \right\} \cdot 2^{O(n)}.
    \]
\end{proposition}
\begin{proof}
    Let $S$ be any $\epsilon$-eligible sumset, and let $A$ satisfying $A+A = S$ be an additive root of $S$. We bound the number of $\epsilon$-eligible sumsets by considering possibilities for $A$.

    We begin with the observation that if $\Vol_{00}(A), \Vol_{11}(A) > 0$, then it must be true that $\Vol_{00}(A), \Vol_{11}(A) \leq \epsilon$. Otherwise, we would have $\Vol_{11}(A+A) = \Vol_{11}(S) > \epsilon$, contradicting our assumption that $S$ is $\epsilon$-eligible. Likewise, we have that if $\Vol_{01}(A), \Vol_{10}(A) > 0$, then $\Vol_{01}(A), \Vol_{01}(A) \leq \epsilon$. We split into cases accordingly.

    \begin{enumerate}
        \item All four cosets of $\{x: x_1 = x_2 = 0\}$ are nonempty: $\Vol_{00}(A), \Vol_{11}(A), \Vol_{01}(A), \Vol_{10}(A) > 0$. 
        
        In this case, we have that $\Vol_{00}(A), \Vol_{11}(A)$, $\Vol_{01}(A), \Vol_{01}(A) \leq \epsilon$. Using \Cref{eq:stirling}, we can then bound the number of possibilities for $A$ (and $S$) by 
        \[
            {2^n\choose \eps 2^n}^4 \leq 2^{4 H(\eps) \cdot 2^n} \cdot 2^{O(n)}.
        \]
        \item Either $\Vol_{00}(A)$ and $\Vol_{11}(A) > 0$, or $\Vol_{01}(A)$ and $ \Vol_{10}(A) > 0$, but not both. 
        
        Here the volume of $A$ on two of the four cosets is at most $\epsilon$, in one other coset it is 0, and in the final coset it may be as large as 1. In this case, again using \Cref{eq:stirling}, the number of possibilities for $A$ (and $S$) is bounded by
        \[
            2^{2^{n}} \cdot {2^n\choose \eps 2^n}^2 \leq 2^{(1+ 2H(\eps)) \cdot 2^n} \cdot 2^{O(n)}.
        \]
        \item Either $\Vol_{00}(A)$ or $\Vol_{11}(A) = 0$, and either $\Vol_{01}(A)$ or $\Vol_{10}(A) = 0$, hence at least two of the four cosets contain no points in $A$.

        Assume without loss of generality that $\Vol_{11}(A) = \Vol_{01}(A) = 0$. This immediately implies that $\Vol_{01}(S) = \Vol_{11}(S) = 0$, as $A+A$ cannot contain points in either coset. (Note that, whichever pair of cosets we choose to zero out, this implies that $\Vol_{11}(S) = 0$ and either that $\Vol_{01}(S) = 0$ or $\Vol_{10}(S) = 0$.) Using the fact that $\Vol_{00}(S) \geq 1-\epsilon$, the number of possibilities for $S$ is bounded by 
        \[
            2^{2^{n}} \cdot {2^n\choose \eps 2^n} \leq 2^{(1+ H(\eps)) \cdot 2^n} \cdot 2^{O(n)}.
        \]
    \end{enumerate}
    Summing the number of $\epsilon$-eligible sumsets covered by each case completes the proof.
\end{proof}

% \begin{proof}
% We consider two separate cases:

% \begin{enumerate}
%     \item $\Vol_{00}(S) = \Vol_{01}(S) = 1$. In this case, as $\Vol_{11}(S) \leq \eps$, we have at most 
%     \[
%         2^{2^n} \times {2^n\choose \eps 2^n} \leq 2^{(1+ H(\eps))2^n}
%     \]
%     possibilities for $S$, as $S$ can be arbitrary on the $(1,0)$-coset and must have relative-density at most $\eps$ in the $(1,1)$-coset. The inequality follows from \Cref{eq:stirling}.
%     \item There exists $x \notin S$ with $x_1 = 0$.
% \end{enumerate}
% \noindent{Case 2: There exists $x \notin S$ with $x_1 = 0$.} In this case, note first that as $\Vol_{11}(S)\leq \eps$, we must have that if $\Vol_{10}(A), \Vol_{01}(A)\neq 0$, then $\Vol_{10}(A) \leq \eps.$ By similar reasoning, we must also have that if $\Vol_{00}(A), \Vol_{11}(A)\neq 0$, then $\Vol_{11}(A) \leq \eps.$ Consider the graph $\Gamma_{\F_2^n}(x)$ as defined in the proof of \Cref{prop:sumset-count}. As $A$ must be an independent set in $\Gamma_{\F_2^n}(x)$, and since $\Gamma_{\F_2^n}(x)$ is a matching, we the number of such sets is at most 
% \[3^{2^{n+1}} \times {2^n\choose \eps 2^n} \times {2^n\choose \eps 2^n} \approx 2^{1.58\times 2^{n+1} + H(\eps)2^{n+1}}.\]
% On the other hand, if one of $\Vol_{10}(A), \Vol_{01}(A)$ is $0$\ldots \red{TODO: Rearrange the casework so that it makes more sense.}   
%\end{proof}

We conclude with the proof of \Cref{prop:sumset-no}.

\begin{proofof}{\Cref{prop:sumset-no}}
    By \Cref{prop:eligible-sumset-count}, the number of $\epsilon$-eligible sumsets is $2^{ (1 + 2H(\epsilon)) 2^n} \cdot 2^{O(n)}$ when $\epsilon < 0.1$. By \Cref{eq:stirling}, the number of \emph{subsets} of $\F_2^{n+2}$ that are $\gamma$-close to a given sumset is 
    \[
        \binom{2^{n+2}}{\gamma 2^{n+2}} = 2^{H(\gamma)2^{n+2}} \cdot 2^{O(n)}.  
    \]
    Thus, by union-bounding over all $\epsilon$-eligible sumsets, we conclude that the number of subsets of $\F_2^{n+2}$ that are $\gamma$-close to any $\epsilon$-eligible sumset is at most
    \[
        2^{ (1 + 2H(\epsilon)) 2^n + H(\gamma) 2^{n + 2}} \cdot 2^{O(n)}. 
    \]

    Choosing $\epsilon = 0.05$ and $\gamma = \epsilon / 4$ gives an upper bound of $2^{1.96 \cdot 2^n}$ subsets of $\F_2^{n+2}$ that are $\epsilon/4$-close to any $\epsilon$-eligible sumset. Since $\calS_{\no}$ is distributed uniformly over $2^{2^{n+1}}$ subsets, the probability that $\bSno\sim \calS_{\no}$ is $(\epsilon/4)$-close to any $\epsilon$-eligible sumset is $2^{-\Omega(2^n)}$.

    We further claim that $\bSno\sim \calS_{\no}$ is always $(\eps/4)$-far from any sumset that is not $\eps$-eligible.~This is just because that we always have 
    %$\$
    %is at least $(\epsilon / 4)$-far from $\calS_{\no}$ with probability $1 - o_n(1)$. To see this, observe that 
    $\Vol_{11}(\bSno) = 0$ and $\Vol_{00}(\bSno) = 1$.
    % \xnote{I think this is always the case, instead of probability $1-o_n(1)$.} % with probability $1-o_n(1)$. 
    On the other hand, any non-$\epsilon$-eligible sumset $S$ has either $\Vol_{00}(S) < 1-\epsilon$ or $\Vol_{11}(S) > \epsilon$ by definition and thus, must be at least $(\epsilon / 4)$-far from $\bSno$. 
    %with probability $1-o_n(1)$.

    Thus with probability at least $1 - o_n(1)$, $\bSno\sim\calS_{\no}$ is $\epsilon/4 = 0.0125$-far from any sumset.
\end{proofof}
%!TEX root = ../testing-sumsets.tex

\section{Refuting Sumsets in the Smoothed Analysis Setting}
\label{sec:sumset-certifications}

In this section we study the smallest size of a \emph{certificate} that a set is \emph{not} a sumset.
Informally, for a set~$S\sse\F_2^n$ a sumset~$0$-certificate is a set~$D\subseteq \F_2^n$ of points such that querying the oracle~$\calO_S$ on every point in~$D$ suffices to prove that~$S$ is not a sumset. The formal definition follows:

\begin{definition}\label{def:0cert}
    A set~$D\subseteq \F_2^n$ is a \emph{sumset~$0$-certificate for}~$S\sse\F_2^n$ if there is no {sumset}~$S'=A+A \subseteq \F_2^n$ for which~$S \cap D = S' \cap D.$
%    \tnote{ Do we prefer $\calO_S \Big\vert_{D}$ to $\calO_S(D)$? The latter seems more natural to me, but I don't know if there's an existing standard. We could also avoid the issue entirely by using $S \triangle D$. {\bf Rocco:} Yeah, maybe it's more natural to phrase the formal definition in terms of $S \cap D$, I wrote it that way up above.}
\end{definition}

Small~$0$-certificates are important objects of study for many property testing problems; for example, consider the classic problem of linearity testing. Since a function $f: \F_2^n \to \F_2$ is linear if and only if $f(x+y)=f(x) + f(y)$ for all $x,y \in \F_2^n$, the property of linearity is characterized by the non-existence of a ``linearity 0-certificate'' of size three.  As is well known, in the seminal work \cite{BLR93} Blum et al.~showed that this is a \emph{robust} characterization, in the sense that a simple sampling procedure which queries random triples $x,y,x+y$ and checks whether they constitute a linearity 0-certificate suffices to distinguish linear functions from functions which are far from being linear.  A similar framework of sampling 0-certificates is at the heart of many other important property testing results such as low degree testing (see e.g.~\cite{AKKLR,jutpatrudzuc04} and many other works) and testing triangle-freeness (see e.g.~\cite{AFKS00,Alon02} and many other works). Of course, testing results of this sort rely on, and motivate the discovery of, structural results showing that functions which are far from having the property in question must have ``many'' ``small'' 0-certificates.

With this motivation, it is natural to study the size of sumset 0-certificates.  Our sumset testing lower bound from \Cref{sec:sumset-lb} suggests that there are sets which are far from being sumsets but which do not have ``many'' ``small'' sumset 0-certificates. In fact, known results imply that for every non-sumset the smallest 0-certificate is of size $\Omega(2^{n/2}/\sqrt{n})$:
% \rnote{We could state, here or elsewhere (intro) that this immediately implies that any one-sided tester for sumsets must make $\Omega(2^{n/2}/\sqrt{n})$ queries. But maybe that undercuts our two-sided tester lower bound. Thoughts?}

\begin{lemma}
Let $S\sse\F_2^n$ be any non-sumset.  Then any sumset $0$-certificate~$D$ for $S$ must have $|D|\geq \Omega(2^{n/2}/\sqrt{n})$.
\end{lemma}
\begin{proof}
This is an immediate corollary of a result due to Alon (Section~4 of \cite{Alon07}), which shows that any subset $T \subseteq \F_2^n$ of size $|T| \geq 2^n - {\frac 1 {4000}} {\frac {2^{n/2}}{\sqrt{n}}}$ is a sumset. It follows that if $|D| < {\frac 1 {4000}} {\frac {2^{n/2}}{\sqrt{n}}}$, then any 0/1 labeling of the points in $D$ is consistent with a sumset (by labeling all points in $\F_2^n \setminus D$ as belonging to the set).
\end{proof}

The previous lemma, which establishes than any $0$-certificate for sumset testing must have size $\Omega(2^{n/2}/\sqrt{n})$, establishes Part (2) of \Cref{thm:smoothed-sumset-refutation}. In the remainder of this section, we prove a matching upper bound (up to a factor of $\poly(n, 1/\epsilon)$) for any set perturbed by a small amount of random noise, thereby establishing Part (1) of \Cref{thm:smoothed-sumset-refutation}:

\begin{theorem}\label{thm:0cert_noised}
    For any set~$S\sse \F_2^n$ and any~$\varepsilon\in (0,\frac{1}{2}]$, there exists a sumset~$0$-certificate for~$\bN_\eps(S) = S\triangle \bR_\varepsilon$ of size~$2^{n/2} \cdot O(n^{1.5}/\epsilon^{1.5})$, with probability $1 - o_n(1)$ over the random draw of $\bR_\eps$. 
    Moreover, such a $0$-certificate can be found efficiently and non-adaptively (by querying $2^{n/2} \cdot O(n^{1.5}/\epsilon^{1.5})$ points) given oracle access to $\bN_\eps(S).$
\end{theorem}

The proof of \Cref{thm:0cert_noised} is given at the end of this section, after several preliminary lemmas. We first observe that for any sumset $S = A+A$, the set $A$ is an independent set in the Cayley graph on $\F_2^n$ with respect to every generator $s \not\in S$:

\begin{lemma}\label{lem:cayley}
    Fix a sumset $A + A = S \subseteq \F_2^n$ and a set $D \subseteq \F_2^n$ with $S \cap D = \emptyset$. Then $A$ is an independent set in $\Gamma_{\F_2^n}(D)$. 
\end{lemma}
\begin{proof}
    Assume for contradiction that~$A$ is not an independent set in~$\Gamma_{\F_2^n}(D)$. This implies the existence of $x, y \in A$ with $x + y = s$ for some $s \in D$. This implies $s \in D \cap (A+A) = D \cap S$, which is a contradiction.
\end{proof}

\begin{lemma}\label{lem:zerocoset}
    Let~$V\subseteq \F_2^n$ be a linear subspace of~$\F_2^n$, and let~$x_1,\ldots, x_{2^{n-\dim(V)}}$ be a choice of representatives from all cosets of~$\F_2^n \setminus V$.
    Let~$A\subseteq \F_2^n$ be an arbitrary subset and denote by~$A_i:= A\cap \left(V+x_i\right)$ the intersection of~$A$ with the~$i$-th coset of~$V$. Then, 
    \[
    (A+A)\cap V = \bigcup_{i=1}^{2^{n-\dim(V)}} (A_i + A_i)
    .
    \]
\end{lemma}
\begin{proof}
    Any element of $A_i + A_i$ belongs to $(A+A) \cap V$, as any two elements of $A_i$ are in the same coset of $V$ by definition.

    It remains to show that any element of $(A+A) \cap V$ is contained in $A_i + A_i$ for some $i \in [2^{n - \dim(V)}]$.  Fix an element $y \in (A+A) \cap V$. Thus $y = a_i + a_j$ for some $a_i \in A_i$ and $a_j \in A_j$. Note that $y \in V$ implies $A_i + A_j = V$, which in turn implies $x_i + x_j \in V$; as $x_i, x_j$ are coset representatives of $V$, this implies that $A_i = x_i + V = x_j + V = A_j$. We thus have that $y \in A_i + A_i$ for some $i$.
\end{proof}

\begin{lemma}
    \label{lemma:not-in-noisedS}
    Let~$S\subseteq \F_2^n$ and~$\varepsilon\in (0,\frac{1}{2}]$. It holds for every~$s\in \F_2^n$ that~$\Pr\left[s\notin \bN_\eps(S)\right]\geq \varepsilon$.
\end{lemma}
\begin{proof}
    Fix $S \subseteq \F_2^n$, $\epsilon \in (0, 1/2]$, and $s \in \F_2^n$. If $s \in S$, then $\Pr[s\notin \bN_\eps(S)] = \epsilon$; if $s \not\in S$, then $\Pr[s\notin \bN_\eps(S)] = 1 - \epsilon \geq \epsilon$.
\end{proof}

The following Lemma follows from the proofs in~\cite{ALON20131232}, we present a self-contained proof for it in \Cref{appendix:indep-number}.
\begin{lemma}
    \label{lem:indep-number}
    The independence number of $\Gamma_{\F_2^n}(\bR_\epsilon)$ is $O(n^2 / \epsilon^2)$ with probability $1 - o_n(1)$.
\end{lemma}

\begin{lemma}
    \label{lem:noisyset-root-bound}
    Let~$V\subseteq \F_2^n$ be a linear subspace of dimension~$d$, and let~$S\subseteq \F_2^n$ be an arbitrary set. The number of different sets~$A\subseteq \F_2^n$ with~$\bN_\epsilon(S) \cap V = (A+A)\cap V$ is with probability at least $1 - o_n(1)$ at most
    \[
        2^{2^{n-d} O(n^3 / \epsilon^2)}.
    \]
\end{lemma}
\begin{proof}
    Fix $S$ and $V$ as in the lemma statement, and let~$x_1,\ldots, x_{2^{n-d}}$ be a choice of representatives from all cosets of~$\F_2^n / V$. We fix $i \in [2^{n-d}]$ arbitrarily and consider possibilities for $A_i \subseteq V + x_i$ such that $A_i + A_i \subseteq {\bN_{\eps}(S)}$.

    Let $\bK$ denote the random set $\bK := V\setminus \bN_{\eps}(S)$.
    % $\bN_\epsilon(S) \cap V$.
    By \Cref{lemma:not-in-noisedS}, $\bK$ includes each point $s \in \mathbb{F}_2^n$ independently with probability at least $\epsilon$.
    %the expected value of $|\bK|$ is at least $\epsilon |V|$.
    %; moreover, since the appearance of each point $s \in \mathbb{F}_2^n$ is independent, we have $|\bK| \geq \epsilon |V| / 2$ with probability $1 - 2^{-\Omega(\epsilon 2^d)}$ via a Chernoff bound.
    By \Cref{lem:cayley}, in order that $A_i + A_i \subseteq {\bN_{\eps}(S)}$.\ it must be true that $A_i$ is an independent set in $\Gamma_{\mathbb{F}_2^n}(\bK)$. 
    %Conditioning on the event that $|\bK| \geq \epsilon |V| / 2$, 
    By \Cref{lem:indep-number}, the independence number of $\Gamma_{\mathbb{F}_2^n}(\bK)$ is $O(n^2 / \epsilon^2)$ with probability $1 - o_n(1)$. (Although $\bK \neq \bR_\epsilon$, \Cref{lem:indep-number} applies because every point appears independently in $\bK$ with probability at least $\epsilon$.)

    Union-bounding over every subset of $V+x_i$ of size $O(n^2 / \epsilon^2)$, we have that with probability $1 - o_n(1)$, the number of possibilities for $A_i$ is at most:
    \[
        \binom{2^d}{ O(n^2 / \epsilon^2) } = 2^{ O(d n^2 / \epsilon^2)}.
    \]
    By \Cref{lem:zerocoset}, any set $A$ satisfies $(A+A) \cap V = \cup_{i \in [2^{n - d}]} (A_i + A_i)$, so the total number of possibilities for $A$ with $A+A = S$ is bounded by 
    \[
        2^{2^{n-d} O(n^3 / \epsilon^2)},
    \]
    as claimed.
\end{proof}

\begin{lemma}
    \label{lem:S-eps-far-bound}
    For any family $\calF$ of subsets of $\F_2^n$, there exists an algorithm that with $O(\frac{2}{\varepsilon}\log |\calF|)$ queries to $\calO_S$ can distinguish between the case that~$S\in \calF$ and the case that~$S$ is~$\varepsilon$-far from every set in~$\calF$ with probability $1 - 1 / |\calF|$.
\end{lemma}
\begin{proof}
    Our algorithm queries $\frac{2}{\epsilon} \log |\calF|$ points selected uniformly at random from $\F_2^n$ and accepts if there exists $F \in \calF$ which is consistent with $S$ on every point. Clearly, the algorithm accepts if $S \in \calF$. It remains to show that we reject if $S$ is $\epsilon$-far from every set in $\calF$.

    In this case, for every set $F$ in $\calF$, we prove $S \neq F$ with probability $\epsilon$ every time we make a query. The chance we fail to prove $S \neq F$ after all queries is thus 
    \[
        (1 - \epsilon)^{\frac{2}{\epsilon} \log{|F|}} \leq e^{- 2\log{|F|}} = 1 / |F|^2
    \]
    by Bernoulli's inequality. Union-bounding over all $F \in \calF$ gives the result.
\end{proof}

\begin{proof}[Proof of \Cref{thm:0cert_noised}]
    Our sumset $0$-certificate for $\bN_\epsilon(S)$ consists of an arbitrary subspace $V$ of dimension $d$, plus an additional $\frac{2}{\epsilon} 2^{n-d} O(n^3 / \epsilon^2)$ points drawn uniformly at random.

    By \Cref{lem:noisyset-root-bound}, with probability at least $1 - o_n(1)$ the number of sets $A$ such that the hypothesis $A+A = S$ is consistent with queries to each element in $V$ is $2^{2^{n-d} O(n^3 / \epsilon^2)}$. By \Cref{lem:S-eps-far-bound}, the remaining random points disprove each possible root $A$ with probability $1 - 2^{-\Omega(n)}$. Setting $d = n/2 + 1.5 \log(n / \epsilon)$ minimizes the size of the certificate at $2^{n/2} \cdot O(n^{1.5}/\epsilon^{1.5})$.
\end{proof}

\subsection{On~$1$-certificates}

In light of \Cref{def:0cert}, it is tempting to consider ``sumset 1-certificates,'' where informally, a sumset 1-certificate for a set $S \subseteq \F_2^n$ is a set $D \subseteq \F_2^n$ such that querying the oracle ${\cal O}_S$ on every point in $D$ suffices to prove that $S$ \emph{is} a sumset.  More precisely,  let us say that a a set $D \subseteq \F_2^n$ is a \emph{sumset 1-certificate for} $S \subseteq \F_2^n$ if $S'$ is a sumset for every set $S' \subseteq \F_2^n$ which satisfies $S \cap D = S' \cap D.$  While this may seem like a natural definition, it turns out that sumset 1-certificates must always be extremely large:

\begin{observation} \label{obs:sumset-1-certificate-large}
If $D$ is a sumset 1-certificate for a set $S \subseteq \F_2^n$, then $|D| \geq 2^{n-1} - O(n^2)$.
\end{observation}
\begin{proof}
There are $2^{2^n - |D|}$ distinct sets $S' \subseteq \F_2^n$ that have $S' \cap D = S \cap D,$ and as shown in \Cref{prop:sumset-count} (see \Cref{appendix:count}) there are at most $2^{2^{n-1} + O(n^2)}$ sumsets in $\F_2^n$. So if $D < 2^{n-1} - O(n^2)$, then there are too many possibilities for $S'$ for all of them to be sumsets.
\end{proof}

% \section*{Acknowledgements}
% X.C. is supported by NSF grants IIS-1838154, CCF-2106429, and CCF-2107187. 
% S.N. is supported by NSF grants IIS-1838154, CCF-2106429, CCF-2211238, CCF-1763970, and
% CCF-2107187. 
% T.R. is supported by \red{TODO}. 
% R.A.S. is supported by NSF grants IIS-1838154, CCF-2106429, and CCF-2211238. 
% O.Z. is supported by \red{TODO}.
% This work was partially completed while some of the authors were visiting the Simons Institute for the Theory of Computing at UC Berkeley. 

\begin{flushleft}
\bibliographystyle{alpha}
\bibliography{allrefs}
\end{flushleft}

\appendix

%!TEX root = ../testing-sumsets.tex

\section{The Number of Sumsets in $\F_2^n$} \label{appendix:count}
\begin{proposition} \label{prop:sumset-count}
    The number of sumsets in $\F_2^n$ is at most 
    \[2^{2^{n-1} + O(n^2)}.\]
\end{proposition}

\begin{proof}

Consider a sumset $S = A+A$; we consider two cases depending on the linear rank of the set $\F_2^n\setminus S$.

\medskip

\noindent{Case 1: $\F_2^n\setminus S$ does not have full rank.} In other words, there exists a vector $v\in\F_2^n$ such that 
\[\abra{x,v} = 1 \qquad\text{implies that}\qquad x \notin S.\]
In particular, we have that $S \sse \{x\in\F_2^n : \abra{x,v} = 1\}$. As there are at most $2^n$ choices for $v$, and for each choice of $v$ there are at most $2^{2^{n-1}}$ choices for $S$, we have that there are at most $2^{2^{n-1} + n}$ many sumsets of this form.

\medskip

\noindent {Case 2: $\F_2^n\setminus S$ has full rank.} In particular, there are $n$ linearly independent vectors \emph{not} in $S$. For $v\in\F_2^n$, observe that the Cayley graph $\Gamma_{\F_2^n}(v)$ is a perfect matching. Next, note that if $v\notin S = A+A$, then $A$ must be an independent set in $\Gamma_{\F_2^n}(v)$. This is because, if $x, y \in A$ with $x + v = y$ then 
\[A + A \ni x + y = x+ x + v = v \notin A+A,\] 
which is a contradiction. As we have $n$ linearly independent vectors not in $S$, it follows that there exists an orthogonal transformation of $\F_2^n$ such that $A$ must be an independent set in the hypercube (where edges are incident to elements of $\F_2^n$ that differ in a single coordinate). As the number of independent sets in $\F_2^n$ is at most $2^{2^{n-1} + O(1)}$ (see for example \cite{galvin2019independent}), and as the number of orthogonal transformations of the hypercube is at most $2^{n^2}$, it follows that the total number of sumsets of this form is at most 
\[2^{2^{n-1} + n^2 + O(1)}.\]

\medskip

Both cases together complete the proof.
\end{proof}

\section{Proof of \Cref{lem:indep-number}}\label{appendix:indep-number}

\begin{definition}
    We say that a set~$A\subseteq \F_2^n$ has \emph{many sums} if~$|A+A|\geq \frac{1}{4} |A|^2$.
\end{definition}

\begin{lemma}\label{alem:manysums}
    Let~$A\subseteq \F_2^n$ be a non-empty set, then there exists~$A'\subseteq A$ such that~$|A'|\geq \sqrt{|A|}$ and~$A'$ has many sums.
\end{lemma}
\begin{proof}
    We construct~$A'$ iteratively in a greedy manner. 
    Starting from~$A'=\{a\}$ for an arbitrary~$a\in A$, as long as there is any~$x\in A\setminus A'$ such that~$\left(A'+x\right)\setminus \left(A'+A'\right) \geq \frac{1}{2} \left(|A'|+1\right)$ then we add~$x$ to~$A'$.

    Let~$a\in A'$.
    We observe that~$|A'+A'|\leq \frac{|A'|\cdot \left(|A'|-1\right)}{2}$, and thus for a uniformly chosen~$\bx\sim A$ we have
    $$\Pr[a+\bx \in A'+A']\leq \frac{|A'+A'|}{|A|} \leq \frac{|A'|\cdot \left(|A'|-1\right)}{2|A|}
    .$$
    Hence, for a uniformly chosen~$\bx\sim U(A)$ the expected number of elements~$a\in A'$ for which~$a+\bx\in A'+A'$ is at most~$|A'| \cdot \frac{|A'|\cdot \left(|A'|-1\right)}{2|A|} = \frac{|A'|-1}{2} \cdot \frac{|A'|^2}{|A|}$. 
    In particular, there exists such an outcome of~$x\in A$.
    If~$|A'|\leq \sqrt{|A|}$ then~$x+a\in A'+A'$ for at most~$\frac{|A'|-1}{2}$ elements~$a\in A'$. Each of the sums~$(x+a)$, for all~$a\in A'$, is unique as~$x$ is fixed. Thus, $\left(A'+x\right)\setminus \left(A'+A'\right) \geq \frac{1}{2} \left(|A'|+1\right)$.
    We conclude that the greedy process would not halt before~$A'$ is of size at least~$\sqrt{|A|}$.
    
    Finally, note that by definition of the greedy process we have
    $$
    |A'+A'| \geq \frac{1}{2}\left(1+2+\ldots+|A'|\right) \geq \frac{1}{4} |A'|^2
    .$$
\end{proof}

\begin{lemma}\label{alem:manysumsIS}
    The probability that~$\Gamma_{\F_2^n}(\bR_\epsilon)$ contains any independent set~$A$ that has many sums and is of size~$|A|\geq \frac{5n}{\varepsilon}$ is at most~$2^{-\Theta(n^2/\varepsilon)}$.
\end{lemma}
\begin{proof}
    Denote by~$k:=\lceil \frac{5n}{\varepsilon}\rceil$.
    For a specific set~$A\subseteq \F_2^n$ of size~$|A|=k$ which has many sums, the probability it is an independent set in~$\Gamma_{\F_2^n}(\bR_\epsilon)$ is at most~$(1-\varepsilon)^{\frac{1}{4}k^2}$.
    The number of sets in~$\F_2^n$ of size~$k$ is~${2^n \choose k}\leq 2^{nk}$, and in particular this is an upper bound for the number of such sets that also have many sums.
    Thus, the expected number of sets~$A$ of size~$k$ with many sums that are independent sets in~$\Gamma_{\F_2^n}(\bR_\epsilon)$, by union bound, is at most
    \[
    2^{nk} \cdot (1-\varepsilon)^{\frac{1}{4}k^2} =
    2^{k\left(n + \frac{1}{4} \log\left(1-\varepsilon\right) k\right)}
    <
    2^{k\left(n - \frac{\varepsilon}{4} k\right)}
    <
    2^{\frac{5n}{\varepsilon} \cdot \left(-\frac{n}{4}\right)}
    =
    2^{-\frac{5}{4\varepsilon}n^2}
    .\]
\end{proof}

\begin{proof}[Proof of \Cref{lem:indep-number}]
    By Lemma~\ref{alem:manysumsIS}, with probability at least~$1-2^{-\Theta(n^2/\varepsilon)}$ there is no independent set in~$\Gamma_{\F_2^n}(\bR_\epsilon)$ of size~$\geq \frac{5n}{\varepsilon}$ with many sums.
    By Lemma~\ref{alem:manysums}, if there is any independent set in~$\Gamma_{\F_2^n}(\bR_\epsilon)$ of size~$\geq \frac{25n^2}{\varepsilon^2}$ then it has a subset of size~$\geq \frac{5n}{\varepsilon}$ which is also an independent set but has many sums, which is a contradiction.
\end{proof}

\end{document}